\documentclass{article}
\usepackage{geometry}                % See geometry.pdf to learn the layout options. There are lots.
\usepackage{graphicx}
\usepackage{stmaryrd}
\usepackage{amssymb}
\usepackage{amsthm, bm}
\usepackage{bbm}
\usepackage{amsmath, cancel, centernot }
\usepackage[mathscr]{eucal}
\usepackage{mathtools}
\usepackage{epstopdf}
\title{Quantum Physics from Number Theory}
\author{Tim Palmer\\ Department of Physics, University of Oxford, UK\\
tim.palmer@physics.ox.ac.uk}
\date{\today}                                          % Activate to display a given date or no date
\makeatletter
\newcommand\be{\@ifstar{\[}{\begin{equation}}}
\newcommand\ee{\@ifstar{\]}{\end{equation}}}
\newcommand\bp{\begin{pmatrix}}
\newcommand\ep{\end{pmatrix}}

\newtheorem*{theorem*}{Theorem}

\newtheorem*{corollary}{Corollary}
\makeatother
\begin{document}
\bibliographystyle{plain}
\maketitle
\abstract{The properties which give quantum mechanics its unique character - unitarity, complementarity, non-commutativity, uncertainty, nonlocality - derive from the algebraic structure of Hermitian operators acting on the wavefunction in complex Hilbert space. Because of this, the wavefunction cannot be shown to describe an ensemble of deterministic states where uncertainty simply reflects a lack of knowledge about which ensemble member describes reality. This has led to endless debates about the ontology of quantum mechanics.

Here we derive these same quantum properties from number theoretic attributes of trigonometric functions applied to an explicitly ensemble-based representation of discretised complex Hilbert states. To avoid fine-tuning, the metric on state space must be $p$-adic rather than Euclidean where $1/p$ determines the fineness of the discretisation. This hints at both the existence of an underpinning fractal state-space geometry onto which states of the world are constrained. In this model, violation of Bell inequalities is a manifestation of this geometric constraint and does \emph{not} imply a breakdown of local space-time causality. 

Because the discretised wavefunction describes an ensemble of states, there is no collapse of the wavefunction. Instead measurement describes a nonlinear clustering of state-space trajectories on the state-space geometry. In this model, systems with mass greater than the Planck mass will not exhibit quantum properties and instead behave classically. The geometric constraint suggests that the exponential increase in the size of state space with qubit number may break down with qubit numbers as small as a few hundred. Quantum mechanics is itself a \emph{singular} limit of this number-theoretic model at $p=\infty$. A modification of general relativity, consistent with this discretised model of quantum physics, is proposed.}

\section{Introduction}
\label{intro}
The quantum revolution was kickstarted by Planck's insight that the energy of an electrically charged oscillator in a cavity containing black-body radiation varies, not continuously, but in discrete increments. This proposal led directly to the development of quantum mechanics, our most successful theory of physics. 

The properties that give quantum mechanics its unique status in physical theory - unitarity, complementarity, non-commutativity, uncertainty and nonlocality - derive from the representation of quantum observables as Hermitian operators on complex Hilbert space. From an experimental point of view there is no reason for questioning this representation. And yet from a theoretical point of view, perhaps there is \cite{NaturePhysics}. Not least, the measurement problem is still felt by many to be unsolved and perhaps unsolvable in quantum mechanics. In addition, the failure of local realism in quantum mechanics, implied by the violation of Bell's inequality \cite{Bell:1964} \cite{CHSH}, is at least in tension with the deterministic causal theory of general relativity. At heart, these conceptual problems arise because the quantum mechanical wavefunction $\psi$ cannot be interpreted as describing an ensemble of states of a deterministically evolving locally causal system. 

In this paper, a different approach is taken. Planck's great insight is applied a second time, to the complex continuum of Hilbert space (if the term `second quantisation' was not already used in physics, it would have formed the title of this paper). That it to say, we discretise Hilbert space \cite{Buniy:2005}. Since the precision at which experiments are performed is necessarily finite, the experiments that confirm quantum mechanics will also confirm a consistent theory in which complex Hilbert vectors and tensor products belong to a sufficiently fine discretisation of Hilbert space. 

Of course, this would achieve nothing if \emph{all aspects} of the discretised model were arbitrarily close to quantum mechanics. However, important differences between quantum mechanics and the discretised model are uncovered. Firstly, the unique properties that characterise quantum mechanical systems (see above) can all, instead, be derived as a simple consequence of number theoretic properties of trigonometric functions, applied to this $\psi$-ensemble \cite{Hance2022Ensemble} description. Here we fix the fineness of the discretisation by the value $1/p$ where $p$ is a large integer. To avoid the derivations from being considered finely tuned (the number-theoretic derivations distinguish between rational and irrational angles), the metric on state space must be $p$-adic rather than Euclidean. This hints at a $p$-adic and hence fractal geometric explanation of quantum nonlocality: not as a physical process that is in tension with the causal space-time geometry of general relativity, but rather as a constraint imposed by fractal state-space geometry.  

In this model there is no `collapse of the wavefunction' and the measurement problem \cite{HanceHossenfelder:2022a} is simply that of describing trajectory clustering - potentially gravitationally induced - in state space. In this framework, quantum mechanics is a \emph{singular} limit of this discretisation at $p=\infty$. This notion is consistent with Hardy's axioms for quantum mechanics - any discretisation, no matter how fine, will violate the Continuity axiom \cite{Hardy:2004}. Indeed, all the number-theoretic results in this paper hold no matter how fine is the discretisation providing $1/p \ne 0$. Nevertheless we estimate a value $p_{\text{max}}$ of $p$ and from this make two predictions: that a limited number of qubits (possibly as few as a couple of hundred) can be generally entangled together, and that objects with mass greater than the Planck mass will not exhibit quantum properties.  

Section \ref{IST} describes the framework for the discretisation of complex Hilbert space and a key number theorem that is used to derive quantum properties from bit string representations of the wavefunction. Section \ref{qubit1} describes the number-theoretic approach to complementarity, non-commutativity, uncertainty, Schr\"{o}dinger evolution and measurement, the quantum mechanical and classical limit of this theory. Section \ref{entangled} describes Bell's theorem and the GHZ state from this discretised perspective. A modification to the field equations of general relativity, consistent with the results above, is proposed in Section \ref{discussion}. This paper is an extension and a simplification of the author's earlier work \cite{Palmer:2020} discretising the Bloch sphere.

\section{Discretisation of Complex Hilbert Space and the Ensemble Representation of the Wavefunction}
\label{IST}

\subsection{The Symbolism of Atomic Measurements}

The title of this sub-section is taken from the sub-title of Schwinger's book \cite{Schwinger} on quantum measurement. In his book, Schwinger attempts to develop quantum mechanics by constructing an algebra of abstract measurement symbols. This notion of measurement outcomes as mere symbols clearly makes some sense: we do not need to know what `spin up' and `spin down' actually means in order to express the spin state of some electron a $|\text{spin up}\rangle+|\text{spin down}\rangle$, any more than we need to know what it means biologically for a cat to be alive or dead, in order to express the state of a cat as $|\text{alive}\rangle+ |\text{dead}\rangle$. In both cases we can simply replace words with symbols so that these states are represented as $|1\rangle + |-1\rangle$. 

This notion of measurement symbolism fits well with the other use of symbols in this paper, to represent a coarse-grain partition of the state space of deterministic dynamical systems. Indeed symbolic dynamics is a mature branch of classical dynamical systems theory \cite{LindMarcus}, and is useful for describing the topology of attractors of dynamical systems. As an example, the unstable periodic orbits (UPOs) which are embedded in the Lorenz attractor \cite{Lorenz:1963} can be represented as finite bit strings (e.g. $1,-1,1,-1,1,1,-1,1$ represents symbolically a UPO of the Lorenz attractor, where `$1$' is a symbol representing the left-hand lobe of the attractor, and `$-1$' is a symbol representing the right-hand lobe). As an example of symbolic dynamics, Ghys \cite{Ghys} has shown that there is a one-to-one correspondence between the symbolic representation of the UPOs of the Lorenz attractor and a symbolic representation of the modular group. 

We discretise complex Hilbert space with these dual notions in mind. Hence, when representing Hilbert states as bit strings, the individual bits are to be thought of as defining underlying deterministic states relative to some coarse-grain partition of some deterministic state space. Long strings of such bits will have well-defined statistical properties in terms of ensemble averages, standard deviations and correlations

\subsection{Discretisation and Niven's Theorem}

Hilbert vectors over the real numbers provide a natural way to represent classical systems under (epistemic) uncertainty. For example, the state of the top card of a well-shuffled face-down deck of cards can be represented as 
\be
|\psi_{\text{ball bearing}} \rangle =\sqrt{\frac{1}{52}} \  | \text{ace of hearts} \rangle + \sqrt{\frac{51}{52}} \ |\text{not ace of hearts} \rangle
\ee
That the squared sum of the components of the vector equals one is consistent with Pythagoras' theorem \cite{Isham}. Born's rule (interpreting squared amplitudes as probabilities) is consistent with the fact that the probabilities of the top card being the ace of hearts and not being the ace of hearts sum to 1. That the squared amplitudes are rational numbers is consistent with the frequentist-based probabilities associated with the finite deck, and it is also consistent with one of Penrose's motivations for developing his model of spin networks \cite{Penrose:2010} where combinatoric rules ensure probabilities are always describable by rational numbers.

Based on the notion of discretisation of Hilbert space, we will be seeking an ensemble description of a generic $K$-qubit quantum state $|\psi_K\rangle$ in the form
\be
\label{hilbert}
a_0 |1,1,\ldots,1\rangle +a_1e^{i \phi_1} |1,1,\ldots, 1,-1 \rangle + a_2e^{i \phi_2} |1,1,\ldots, -1, 1\rangle +\ldots a_{2^K} e^{i \phi_{2^K}} |-1,-1,\ldots -1,-1\rangle
 \ee
where $\sum_0^{2^K} a^2_j=1$. Motivated in part by the discussion above we look for a discretisation of the form 
 \be
 \label{rational1}
 a^2_j = \frac{m_j}{2N} \in \mathbb Q, 
 \ee
where $N \in \mathbb N$ and $0\le m_j\le  2N$. The factor 2 is linked to a specific representation of quaternionic operators discussed in Section \ref{singlequbit}. 

A key difference between the Hilbert space representation of a quantum state and the Hilbert space representation of a classical state under uncertainty, is the essential role that complex numbers play in the former \cite{MingChen}. We seek a discretisation of complex phases of the form
\be
\label{rational2}
 \frac{\phi_j}{2 \pi} = \frac{n_j}{4N}\in \mathbb Q
 \ee
with $1 \le n_j \le 4N$ where again the factor of 4 is linked to a representation of quaternionic operators discussed in Section \ref{singlequbit}. In this way $\{\phi_j\}$ is the finite multiplicative group of $4N$th roots of unity. Importantly for the development below, this is homomorphic to the group of cyclic permutations of the form $(1,2, \ldots 4N)$. 

As an example, in quantum mechanics a single qubit 
 \be
\label{qubit}
|\psi_1(\theta, \phi)\rangle = \cos\frac{\theta}{2} |1\rangle + \sin\frac{\theta}{2} e^{i \phi} | -1 \rangle
\ee
can be represented by a point on the continuum Bloch sphere, where $\theta$ denotes colatitude and $\phi$ longitude relative to some north pole $z$. In the discretised framework, if (\ref{qubit}) satisfies the rationality constraints (\ref{rational1}) and (\ref{rational2}) then
\be
\label{rational3}
\sin^2 \frac{\theta}{2}=\frac{m}{2N}; \ \ \ \  \frac{\phi}{2\pi}=\frac{n}{4N}
\ee
where $0 \le m \le 2N$, $0 \le n \le 4N$. From (\ref{rational3}), $\cos \theta=1- 2\sin^2 \theta/2$ must be rational. 

Let $S_z(\theta, \phi)$ ($S_z$ for short) denote a skeleton of points on the unit 2-sphere (the discretised Bloch sphere) with north pole at $z$, satisfying (\ref{rational3}). A critical property of $S_z$, exploited repeatedly below, is the number theoretic incommensurateness between rational angles and angles with rational cosines, referred to as Niven's Theorem:
\begin{theorem*}
\label{niven}
 Let $\phi/2\pi \in \mathbb{Q}$. Then $\cos \phi \notin \mathbb{Q}$ except when $\cos \phi =0, \pm \frac{1}{2}, \pm 1$. \cite{Niven, Jahnel:2005}
\end{theorem*}

In the remainder of this section, an explicit ensemble-based representation of the $K$-qubit wavefunction is developed Most discussion is devoted to $K=1$. From there, the generalisation to higher qubit states is relatively straightforward. 

\subsection{Single Qubits}
\label{singlequbit}

To start, consider four bit strings $A$, $B$, $C$ and $D$, each of $N$ elements drawn from $\{1,-1\}$ but otherwise arbitrary. As discussed in more detail in Section \ref{evolution}, `$1$' and `$-1$' will be used as symbolic representations of some partition of state space. As discussed in Section \ref{sizeofN}, $N$ is a positive integer inversely dependent on the energy of the quantum system under study. 

The fundamental notion underpinning the construction in this section is the close link between quaternions and rotations in physical space. In particular, with $||$ denoting concatenation, let
\be
\mathcal E = A\;||\;B\;||\;C\;||\;D
\ee
denote a $4N$ element bit string and define
\begin{align}
\label{quat}
i_1 (\mathcal E) &= B\ || -A \;|| -D\; || \ \ \ C; \nonumber \\
i_2 (\mathcal E) &= C\ || \ \ \ D\;|| -A \;|| -B; \nonumber \\
i_3 (\mathcal E) &= D\ || -C\;|| \ \ \ B\; || -A
\end{align}
where the negation of a bit string is obtained by negating all individual bits in the string. It is straightforwardly shown that the operators $i_1$, $i_2$ and $i_3$ are quaternionic, i.e. satisfy
\be
 i^2_1 (\mathcal E) =  i^2_2 (\mathcal E)= i^2_3 (\mathcal E)=-\mathcal E;  \ \ \ \ 
 i_2  \circ i_1 (\mathcal E) = i_3 (\mathcal E)
 \ee
We now construct a 1-parameter family of bit strings $i_1^{(m)}(\mathcal E)$ that monotonically interpolates between $\mathcal E$ and $i_1(\mathcal E)$ as $m$ goes from $0$ to $N$. For example, with $N=3$, this is achieved with
\begin{align}
\label{N3}
i_1^{(0)}(\mathcal E)= \mathcal E=& \ a_1 a_2 a_3\  ||\ \ \ \ \ \ b_1 b_2 b_3\ \ \ \ \ \ ||\ \ \ \ \ \ \ c_1 c_2 c_3 \ \ \ \ \ \ \ || \ d_1 d_2 d_3 \nonumber \\
 i_1^{(1)}(\mathcal E) = &\ a_1a_2b_3\ ||\ \ \ \ b_1b_2 -a_3\ \ \ \ \ ||\ \ \ \  c_1c_2-d_3\ \ \ \ \  ||\ d_1d_2c_3 
 \nonumber \\
  i_1^{(2)}(\mathcal E) = &\ a_1b_2b_3\ ||\ \ b_1-a_2 -a_3\ \ \ ||\ \ \  c_1-d_2-d_3\ \ \  ||\ d_1c_2c_3
 \nonumber \\
   i_1^{(3)}(\mathcal E) =i_1(\mathcal E)= &\ b_1b_2b_3\ ||\ -a_1-a_2 -a_3\ ||\ -d_1-d_2-d_3\ ||\ c_1c_2c_3 
 \end{align}
 where the correlation between  $i_1^{(m)}(\mathcal E)$ and $\mathcal E$ equals $1-m/3$ and the correlation between $i_1^{(m)}(\mathcal E)$ and $i_1(\mathcal E)$ equals $m/3$. More generally, define partial concatenation operators
\be
B {_{||}^m} C =\{\underbrace{b_1, b_2, \ldots b_{N-m}}_{N-m} \underbrace{c_{N-m+1}, c_{N-m+2}, \ldots c_N}_m\}
\ee
so that $B{_{||}^0}C=B$, $B{_{||}^N}C=C$. Define the following generalisations $i_1^{(m)}$ of $i_1$ for $0\le m\le 2N$:
\begin{align}
\label{im}
&i_1^{(m)}(\mathcal E)= \nonumber \\
&\ \ \ \ \ \ \  \ \ A {_{||}^m} B\ \ \ \ \ \ \ \  \ ||\ \ \ \ \ \ B{_{||}^m}(-A)\ \ \ \   ||\ \ \ C{_{||}^m}(-D)\ \ \ \ \ \ \ \ || \ \ \ \ \ \ \ \ \ D{_{||}^m}C
\ \ \ \ ; 0\le m \le N \nonumber \\ 
 &\ \ \ \ \ \ B{_{\ \ \ ||}^{(m-N)}}(-A) \ ||\ (-A){_{\ \ ||}^{m-N}}(-B) \ || \ (-D){_{\ \ ||}^{m-N}}(-C)\ || \ \ C{_{\ \ ||}^{m-N}}(-D)\ \ \ ;N \le m\le 2N 
 \end{align}
Note that, $i_1^{(0)}(\mathcal E)= \mathcal E$, $i^{(N)}_{1} (\mathcal E)=i_1(\mathcal E)$,  $i^{(2N)}(\mathcal E)=-\mathcal E$. The correlation between $i_2^{(m)}(\mathcal E)$ and $\mathcal E$ is equal to $1-m/N$. 

We now let $i_1^{(m,n)}(\mathcal E)$ denote the result of applying the cyclic permutation $(1,2,3\ldots 4N)$ $n$ times to $i_1^{(m)}(\mathcal E)$. Hence, from (\ref{N3}), for $N=3$
\begin{align}
i_1^{(0,0)}(\mathcal E)= & \ a_1 a_2 a_3\  ||\ \ \ b_1 b_2 b_3\ \ \ \ \ \ ||\ \ \ \ \ \ c_1 c_2 c_3 \ \ \ \ \ || \ \ \  d_1 d_2 d_3 \nonumber \\
 i_1^{(1,1)}(\mathcal E) = &\ c_3a_1b_2\ \ ||\ \ \ \ b_3b_1 b_2\ \ \ \ \ ||\ \ \ \  -a_3 c_1c_2\ \ \  ||\ -d_3d_1c_2 
 \end{align}
and so on. 

Finally, in the discretised representation of complex Hilbert space, we associate the single qubit wavefunction $|\psi_1(\theta, \phi)\rangle$ (see (\ref{qubit})) with the length $4N$ bit string
\be
\label{bloch}
\mathcal B (\theta, \phi)=i_1^{(m,n)} (\mathbbm 1).
\ee
where
\be
\mathbbm 1= I \ ||\  I \ || \ I \ || \ I , \ \ \ \ I =\{ \underbrace{1,1,1, \ldots, 1}_N \}
\ee
In (\ref{bloch}): i) $\sin^2 \theta/2 = m/2N$ (the fraction of $-1$s in $\mathcal B(\theta, \phi)$) and $\cos \theta = 1-m/N$ (the correlation between $\mathcal B(\theta, \phi)$ and $\mathcal B(0,0)$); ii) $\phi= 2\pi n/4N$. The property 
\be
\label{complex}
  i_3 \; (\mathcal B(\frac{\pi}{2},0))=i_3 \circ i_1(\mathbbm 1)= -i_2(\mathbbm 1) = i_1^{(N,N)}(\mathbbm 1)=\mathcal B (\frac{\pi}{2},\frac{\pi}{2}) 
\ee
describes the essential complex nature of the bit-string states (and the associated skeleton $S_z$) in this discretised model - operating on a bit string at the equator of the discretised Bloch Sphere by the square root of minus operator $i_3$ is equivalent to a rotation in longitude by an angle $\pi/2$ about the $z$ axis. 

Just as the quantum wavefunction is defined modulo a global phase transformation, so, in this discretised ensemble representation, the statistical properties of $\mathcal B (\theta, \phi)$ are preserved under a global permutation of bits - one that is consistently applied to all $\mathcal B (\theta, \phi)$. We can represent this global phase transformation by replacing $\mathbbm 1$ in (\ref{bloch}) with a more general $\mathcal E$. 

This construction can be generalised by including the notion of a `null measurement' \cite{Hardy:2004}. For example, we may have two (or more) detectors, each represented by a symbolic label, but a particle may be incident on neither (or none) of them. This generalisation can be achieved  by appending the null symbol $X$ to $i_1^{(m)}(\mathbbm 1)$, $n_X$ times.  With such a generalisation, $\cos^2 \theta/2 = 1-m/2N$ equals the fraction of $1$s amongst all non-null states in $\mathcal B(\theta, \phi)$ whilst $\phi= 2\pi n/p$ where $p=4N+n_X$. If $n_X$ is odd and $p$ not divisible by 3, then the exceptions $\cos \phi = 0, \pm 1/2, -1$ to Niven's theorem can be ignored (which we henceforth do). 

\subsection{Multiple Qubits}
\label{K}
In quantum mechanics, let $|\psi_1\rangle$, $|\psi_2\rangle$ denote two general $K-1$ qubit states in quantum mechanics. Then a general $K$ qubit state can be written as
\be
\label{psin}
\cos \frac{\theta}{2} |\psi_1\rangle |1\rangle + \sin \frac{\theta}{2} e^{i \phi} |\psi_2\rangle|-1\rangle
\ee
Here the number $2^K-1$ of complex degrees of freedom in an $K$-qubit Hilbert vector increases exponentially as $1,3,7,15 \ldots$ for $K=1,2,3,4,\ldots$. 

A general $K$-qubit Hilbert state in the discretised model is associated with $K$ bit strings, each string comprising $2^{K+1}N$ bits. The state can be defined inductively. Let $\mathcal B_1 (K-1)$, $\mathcal B_2(K-1)$ denote two ordered collections of $K-1$ bit strings each of length $2^K N$ and let $\mathcal B_1 \oplus \mathcal B_2$ denote an ordered collection of $K-1$ bit strings obtained by concatenating the $j$th bit string of $\mathcal B_1$ with the $j$th bit string of $\mathcal B_2$, so that each bit string of $\mathcal B_1 \oplus \mathcal B_2$ has length $2^{K+1}N$. Then, consistent with (\ref{psin}), the $K$-qubit discretised Hilbert state is given by the $K-1$ bit strings $\mathcal B_1(K-1) \oplus \mathcal B_2(K-1)$, and the single length-$4N$ bit string $i_1^{(m_0,n_0)}(\mathbbm 1)$ concatenated $2^{K-1}$ times, i.e., 
\begin{align}
& \ \ \ \ \ \ \ \ \ \ \ \mathcal B_1(K-1)\ \oplus  \  \mathcal B_2(K-1)  \nonumber \\
& i_1^{(m_0,n_0)}(\mathbbm 1) \ || \ i_1^{(m_0,n_0)}(\mathbbm 1) \ || \ldots \ ||\ i_1^{(m_0,n_0)}(\mathbbm 1)
\end{align}
For example, a general $K=3$-qubit state comprises the 3 bit strings
\begin{align}
\label{alicebobbits2}
& \ \ \overbrace{i_1^{(m_2, n_2)} (\mathbbm 1)\ || \  i_1^{(m_3, n_3)}(\mathbbm 1)} \ \ \ \ ||
 \ \ \ \ \ \overbrace{i_1^{(m_5, n_5)} (\mathbbm 1)\ || \  i_1^{(m_6, n_6)}(\mathbbm 1)} \ \ \ 
 \nonumber \\
&\ \  \underbrace{i_1^{(m_1, n_1)} (\mathbbm 1)\ || \  i_1^{(m_1, n_1)}(\mathbbm 1)}\ \ \ \ ||
\ \ \ \ \ \underbrace{i_1^{(m_4, n_4)} (\mathbbm 1)\ || \  i_1^{(m_4, n_4)}(\mathbbm 1)} \ \ \  \nonumber \\
& \ \ i_1^{(m_0, n_0)} (\mathbbm 1)\ || \ \ \  i_1^{(m_0, n_0)}(\mathbbm 1) \ \ \ || \ 
 \ \ \ i_1^{(m_0, n_0)} (\mathbbm 1)\ \ \ || \  i_1^{(m_0, n_0)}(\mathbbm 1) \ \ \ 
\end{align}
with 14 degrees of freedom (corresponding to 7 phase degrees of freedom and 7 amplitude degrees of freedom in quantum mechanics). Here the two pairs of bit strings in between the upper and lower braces are generic 2-qubit states. 

The entangled Bell state is discussed in particular in Section \ref{Bell}

\section{The Physical Interpretation of Single Qubit Physics}
\label{qubit1}

\subsection{Complementarity}
\label{MachZ}

In quantum mechanics, systems can exhibit wave-like properties or particle-like properties, but not both simultaneously. As shown in Fig \ref{MZ}, these properties are manifest in a Mach-Zehnder (MZ) interferometer (Fig \ref{MZ}) with the second half-silvered mirror in place, or with the second half-silvered mirror removed, respectively. %..........................................................................................................................................
\begin{figure}
\centering
\includegraphics[scale=0.3]{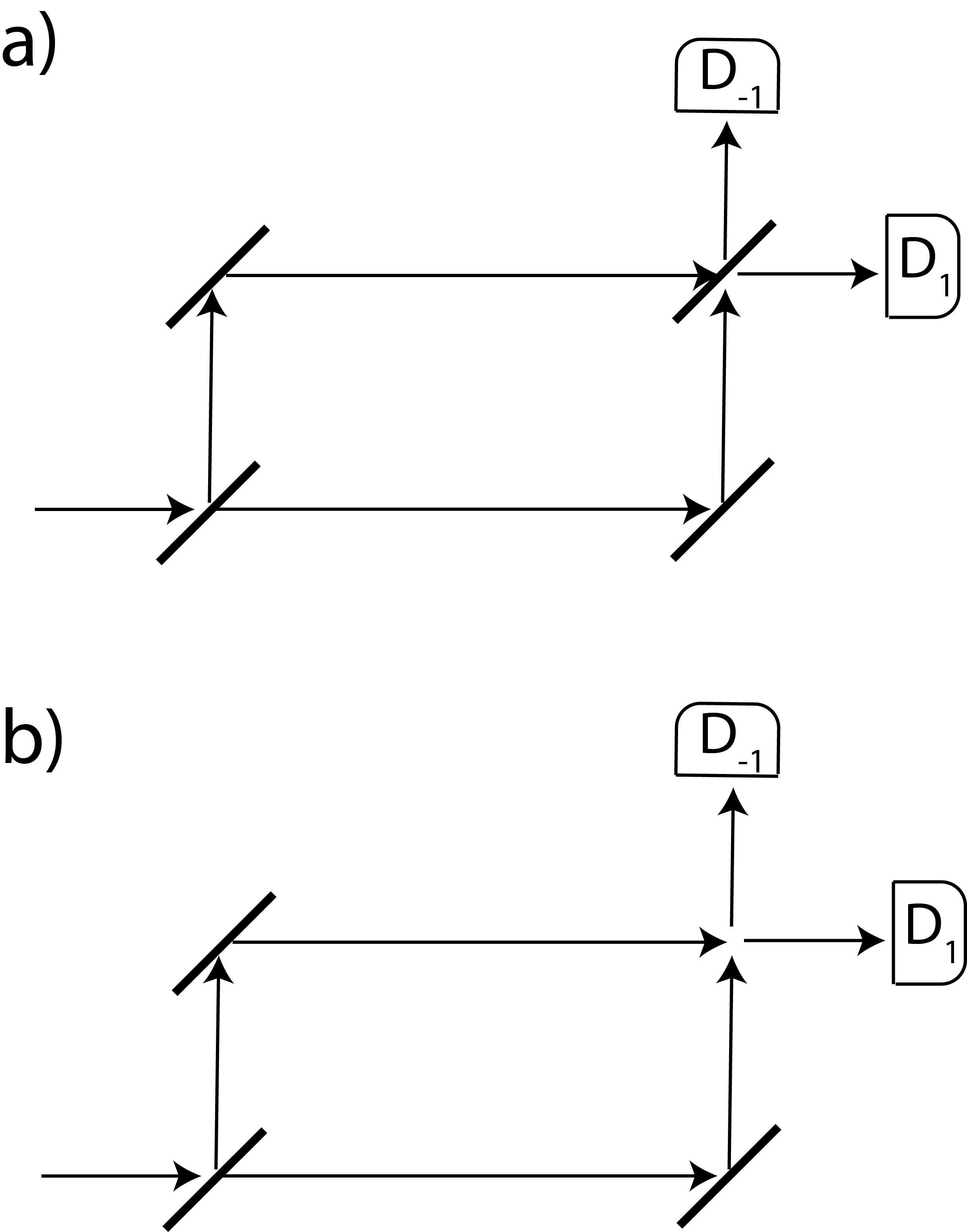}
\caption{\emph{a) A Mach-Zehnder interferometer with the second half-silvered mirror in place, performing an interferometric experiment. b) A Mach-Zehnder interferometer with the half-silvered mirror removed, performing a which-way experiment. In both cases, the nominal lengths of the interferometer arms are the same. However, the exact lengths must necessarily differ. Conversely, for number theoretic reasons, it is impossible for the experimental outcomes in a) and b) to be simultaneously defined, i.e., for the same exact arm lengths, consistent with the quantum notion of complementarity.}}
\label{MZ}
\end{figure}
%.......................................................................................................................................... 
The quantum state of a photon passing through an MZ interferometer can be written as
\be
\label{psi1}
|\psi\rangle= \cos \frac{\phi}{2} |1\rangle + \sin \frac{\phi}{2} |-1\rangle
\ee
where $\phi$ denotes a phase difference associated with the two arms of the interferometer. Here $|1\rangle$ and $|-1\rangle$ are the measurement eigenstates associated with the detectors $D_1$ and $D_{-1}$, respectively. If we now remove the second half-silvered mirror, then the quantum state becomes 
\be
\label{psi2}
|\psi\rangle= \frac{1}{\sqrt 2} ( |1\rangle + e^{i \phi} |-1\rangle)
\ee

As discussed, in the discretised model, $|\psi\rangle$ in (\ref{psi1}) is represented by a bit string $\mathcal B(\phi, 0)$ where the fraction of $1$s is equal to $\cos^2 \phi/2$ (and the fraction of $-1$s is equal to $\sin^2 \phi/2$), providing $\cos^2 \phi/2$ is rational, and of the form $m/2N$. By varying $\phi$ from $0$ to $\pi$ in discrete steps which ensure the rationality condition is satisfied, then the fraction of $1$s varies from $1$ to $0$ in steps which can be as small as we like, for large enough $N$.  In the discretised model, the symbol `$1$' labels a deterministic world where the detector $D_1$ registers a particle, and the symbol `$-1$' labels a deterministic world where the detector $D_{-1}$ registers a particle. How to define `measurement' objectively in this discretised model is discussed in Section \ref{evolution}. By contrast, $|\psi\rangle$ in (\ref{psi2})  is represented by a bit string $\mathcal B(\pi/2, \phi)$ where, for all $\phi$, the fraction of $1$s (and $-1$s) is equal to 1/2 and $\phi$ is a rational angle. 

Now when an experimenter sets up such an experiment, of course he or she cannot control \emph{directly} whether $\phi$ or $\cos \phi$ is rational. All an experimenter can do is set $\phi$ to within some nominal accuracy. We imagine some small interval of size $\epsilon >0$ which defines the nominal accuracy of the experimenter's ability to set $\phi$. The theoretical structure of the model, not some incredible experimental finesse, ensures that the appropriate rationality constraints are obeyed. 

Imagine performing interference measurements on one ensemble of particles with the second half silvered mirror in place, and which-way measurements on a second ensemble of particles with the second half silvered mirror removed. Suppose that $\phi$ is nominally the same (say $\phi=0$) for both ensembles of measurements. Then the exact values $\phi^*$ will not be the same for both ensembles: for the first ensemble, each $\phi^*$ would be drawn from the subset of exact values where $\cos \phi^*$ is rational, whilst for the second ensemble, each $\phi^*$ would be drawn from the subset where $\phi^*$ was itself rational.  Indeed, we can, without loss of generality, assume that each particle is uniquely labelled by an exact $\phi^*$ value. This will be a pervasive theme in the discussion below, that a particle's `hidden variables' can be associated with some exact (but unknown) measurement setting $\phi^*$. Of course this means that, at the time a particle is emitted from a source, its hidden variables are not localised within the particle in question. However, such a situation is no different to describing a human by some lifetime achievement (e.g., whether they won a Nobel Prize or not). The data that determines such a lifetime achievement is manifestly not localised within the baby's DNA at the time of birth, but nevertheless exists on a spacelike hypersurface through the birth event. 

Suppose an interference measurement was performed on a particular particle and consider the hypothetical counterfactual measurement where a which-way measurement was performed \emph{on the same particle}, i.e., keeping the hidden variable fixed. By construction, such an hypothetical experiment is defined by removing the half-silvered mirror but keeping fixed the \emph{exact} lengths of the interferometer arms associated with the real experiment. It is clearly a matter for the model equations as to whether such a hypothetical experiment is consistent with the laws of physics: as mentioned, the experimenter does not have control over the exact measurement settings and therefore cannot ensure that the exact length of the arms remain identical (a passing gravitational wave would cause them to differ) from one run to the next. Indeed, this simultaneous measurement is internally inconsistent: by Niven's theorem it is impossible for $\phi$ to be simultaneously a rational angle and have a rational cosine (above we have eliminated all exceptions to Niven's theorem). Hence, by number theory, it is impossible for a quantum system to have wave-like properties ($|\psi_1\rangle$) and particle-like properties ($|\psi_2\rangle$) simultaneously.

To prevent simultaneous knowledge of the wave and particle properties of a quantum system in their deterministic toy-model of quantum interference, Catani et al \cite{Catani:2022} introduce by assumption an `epistemic obstruction'. We do not need to \emph{postulate} such an obstruction here; it is a \emph{consequence} of the particular discretisation of complex Hilbert Space used. Indeed, as discussed in Section \ref{Bell}, the number-theoretic basis for this epistemic obstruction is identical to that which explains the violation of Bell's inequality - Niven's theorem. From this perspective, we agree with Feynman \cite{Feynman}, and Hance and Hossenfelder \cite{HanceHossenfelder:2022b} that quantum interference is the only real mystery in quantum physics and disagree with Catani et al who claim that it isn't.

\subsection{Non-commutativity}

We analyse a sequential Stern-Gerlach experiment (Fig \ref{SG}a) in the discretised framework. In a particular run of the experiment, a spin-1/2 particle passes through a Stern-Gerlach apparatus ($SG_1$) oriented in the nominal $a$ direction. The spin-up output channel of $SG_1$ is fed into a second Stern-Gerlach apparatus ($SG_2$) oriented in the nominal $b$ direction. The spin-up output channel of $SG_2$ is fed into a third Stern-Gerlach apparatus ($SG_3$) oriented in the nominal $c$ direction. We can suppose that an experiment comprises many individual runs, all with the same nominal orientations, but different exact orientations. 

%..........................................................................................................................................
\begin{figure}
\centering
\includegraphics[scale=0.5]{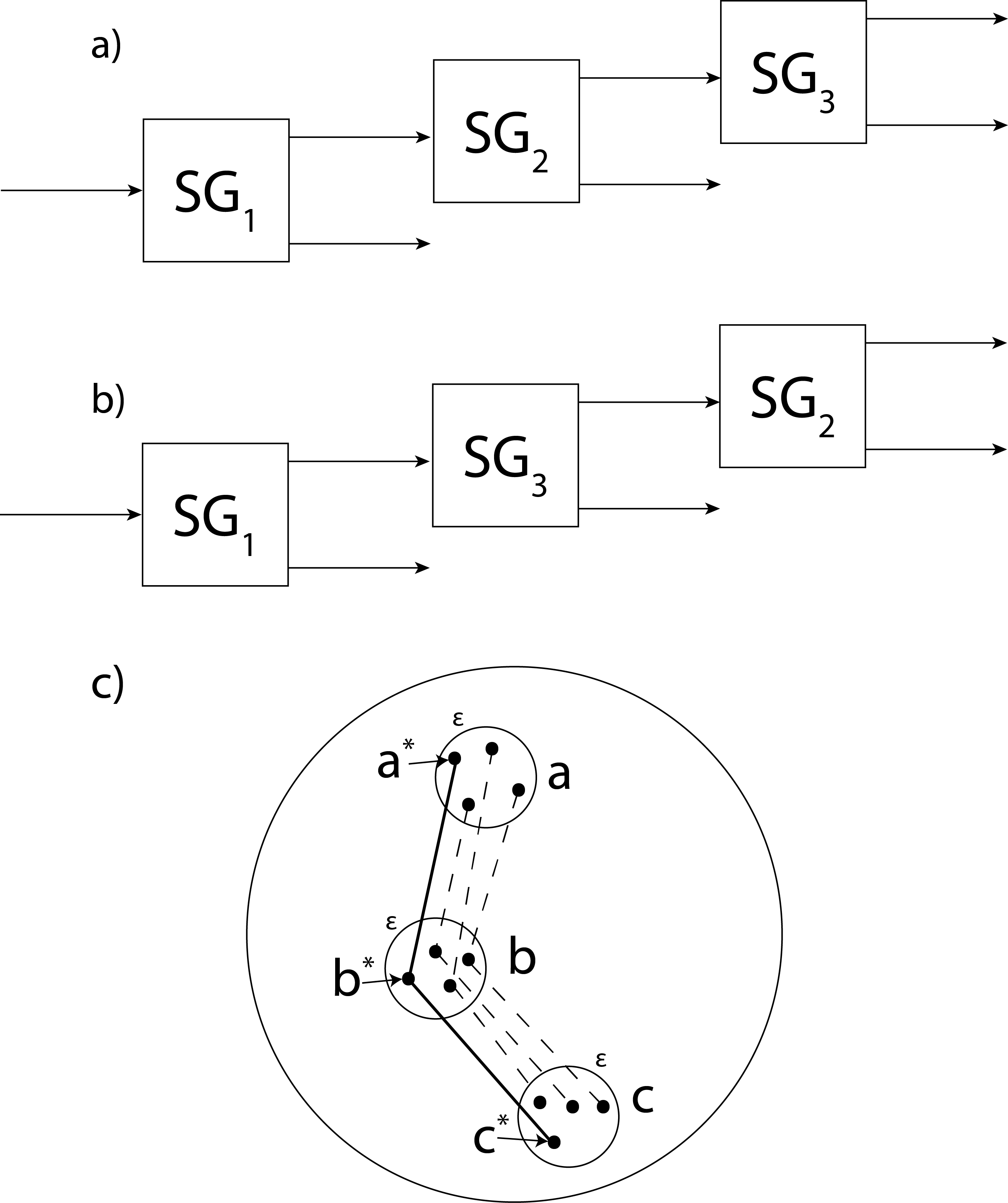}
\caption{\emph{a) A sequential Stern-Gerlach experiment. b) As a) but with the last two Stern-Gerlach devices swapped. c) 4 possible triples of exact orientations for SG1, SG2 and SG3 in a) are shown, consistent with the nominal orientations of the Stern-Gerlach devices and with the rationality constraints for experiment a). One set of triples, shown as $a^*$, $b^*$ and $c^*$, labels a particle's set of hidden variables. By the Impossible Triangle Corollary, these constraints make experiment b) impossible when applied to the particle with the same hidden variables as in a). Specifically, the cosine of the relative angle between $a^*$ and $c^*$ cannot be rational when the cosines of relative angles between $a^*$ and $b^*$ and between $b^*$ and $c^*$ are rational. In this way, the quantum notion of non-commutativity can be explained from the properties of numbers.}}
\label{SG}
\end{figure}
%.......................................................................................................................................... 

For any particular run of this experiment, there exist 3 \emph{exact} orientations referred to as $a^*$, $b^*$ and $c^*$ respectively. In the discretised model, the cosines of the relative orientations between $a^*$ and $b^*$ and between $b^*$ and $c^*$ must be rational. In Fig \ref{SG} c) is shown 4 (of potentially many) different possible triples of exact orientations satisfying these rationality constraints. They correspond to 4 different runs of $SG_1-SG_2-SG_3$, i.e., four different particles with four different \emph{exact} measurement orientations. As discussed further below, these  \emph{exact} measurement settings can be used to define a particle's unique hidden variables.

As with the Mach-Zender experiment where we performed an interference measurement in one run, and a which-way measurement in another run, so here it is possible to perform $SG_1-SG_2-SG_3$ in one run and $SG_1-SG_3-SG_2$ in a second run (Fig \ref{SG}b). In the second run, the cosines of the relative orientations between the exact settings of $SG_1$ and $SG_3$ and between the exact settings of $SG_3$ and $SG_2$ must both be rational. For this second run (unlike the first) we do \emph{not} require that the cosine of the relative orientation between $SG_1$ and $SG_2$ is rational. 

The essential nature of non-commuting quantum observables lies in the impossibility of \emph{simultaneously} performing these two sequential SG experiments a) and b) \emph{on the same particle}, i.e. keeping the particle's hidden variables fixed. According to the definition above, keeping the hidden variable fixed in going from a) to b) - from the real-world run to a counterfactual run - we keep $a^*$, $b^*$ and $c^*$ fixed. But now we run into a number theoretic inconsistency. This is a corollary of Niven's Theorem, referred to as the Impossible Triangle Corollary. 

\begin{corollary}
Let $\triangle{a^*c^*c^*}$ be a non-degenerate triangle on the unit sphere with rational internal angles such that the cosines of the angular distances between $a^*$ and $b^*$ and between $b^*$ and $c^*$ are both rational. Then the cosine of angular distance between $a^*$ and $c^*$ cannot be rational.   \end{corollary}
\begin{proof}
Assume otherwise. Use the cosine rule for $\triangle (a^*b^*c^*)$,
\be
\label{cosinerule}
\cos \theta_{a^*c^*}= \cos \theta_{a^*b^*} \cos \theta_{b^*c^*} + \sin \theta_{a^*b^*} \sin \theta_{b^*c^*} \cos\phi_{b^*}
\ee
where $\theta_{a^*c^*}$ denotes the angular distance between $a^*$ and $c^*$ on the unit sphere, etc and $\phi_{b^*}$ is the internal angle at the vertex $b^*$. Since $\cos\theta_{a^*b^*}$ and $\cos \theta_{b^*c^*}$ are both rational, then from (\ref{cosinerule}), $\sin \theta_{a^*b^*} \sin \theta_{b^*c^*} \cos \phi_{b^*}$ must be rational. Squaring, $(1-\cos^2 \theta_{a^*yb^*})(1- \cos^2 \theta_{b^*c^*}) \cos^2 \phi_{b^*}$ must be rational. Again, since $\cos \theta_{a^*b^*}$ and $\cos \theta_{b^*c^*}$ are both rational, $\cos^2 \phi_{b^*}$ and hence $\cos 2 \phi_{b^*}$ must be rational. Since $\triangle (a^*b^*c^*)$ is non-degenerate $\cos^2 \phi_{b^*}$ must be irrational by Niven's theorem. Hence $\cos \theta_{a^*c^*}$ must be irrational. 
\end{proof}
As discussed at the end of Section \ref{singlequbit}, we assume the bit strings include null measurements to exclude exact multiples of $\pi/4$ and $\pi/3$ in Niven's theorem. 
 
As illustrated in Fig \ref{SG} c), we label one of four possible triples of exact points where both  $\cos \theta_{a^*b^*}$ and  $\cos \theta_{b^*c^*}$ are rational. For this, $\cos \theta_{a^*c^*}$ rational. Hence the hypothetical experiment in Fig \ref{SG} b), counterfactual to the real-world experiment in Fig \ref{SG} a) and using the hidden variable as in a), is not consistent with the rationality conditions of the discretised model. We can therefore explain the non-commutativity of spin from properties of numbers in a discretised model of quantum physics. 

\subsection{Uncertainty}
\label{uncertainty}

%..........................................................................................................................................
\begin{figure}
\centering
\includegraphics[scale=0.3]{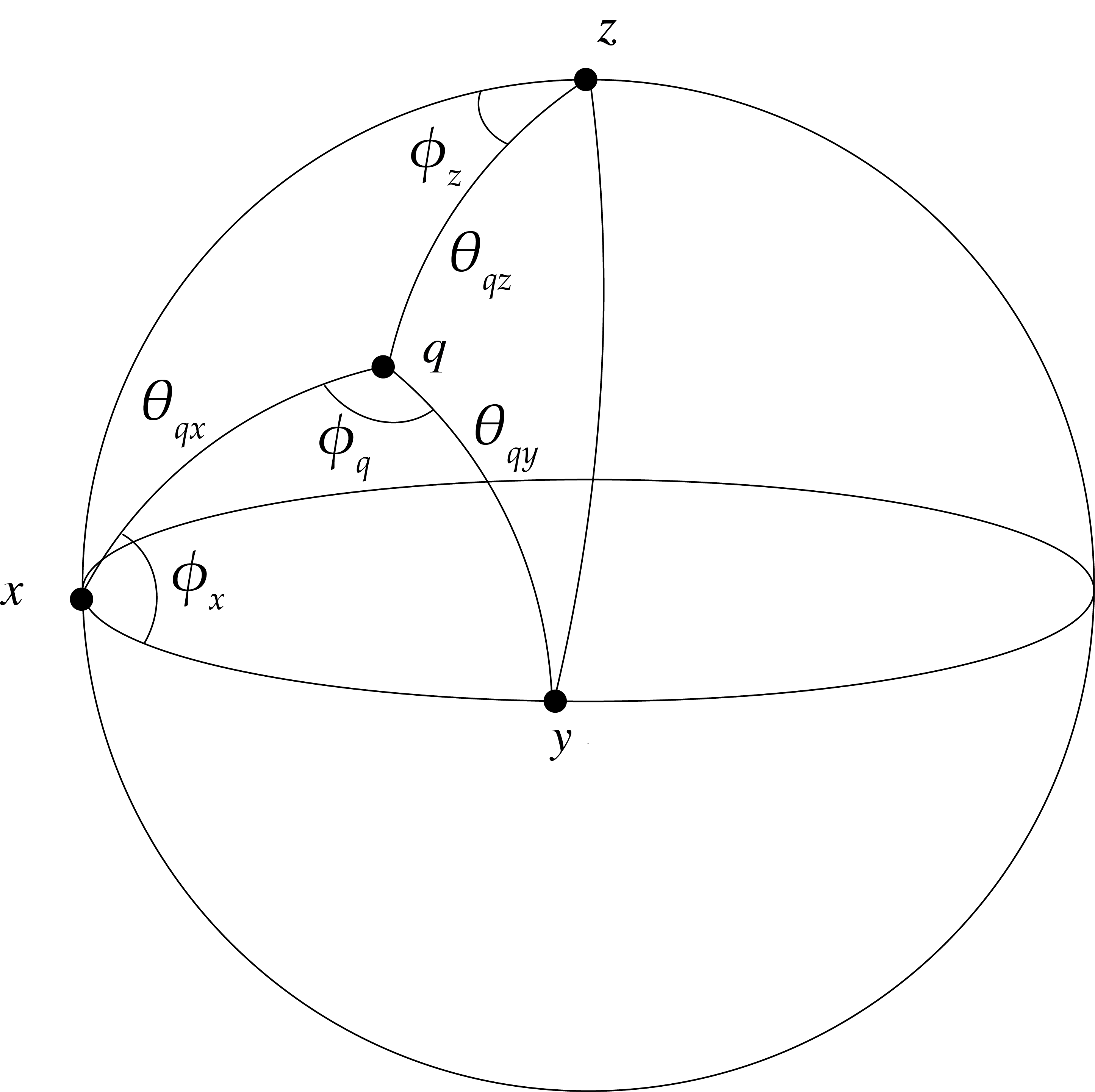}
\caption{\emph{Analysis of the properties of spherical triangles formed from the point $q$ and the nominally orthogonal triple of points, $z$, $x$ and $y$ provides a number-theoretic representation of the Uncertainty Principle, both statistical and ontological.}}
\label{sphtri}
\end{figure}
%..........................................................................................................................................

Consider an exact point $q$ on the Bloch sphere lying on a skeleton with north pole at $z$, where $x$ and $y$ are nominally orthogonal points on the equator (Fig \ref{sphtri}). Using the sine rule for the spherical triangle $\triangle(q, x, y)$ to nominal accuracy
\be
\frac{\sin \phi_x}{\sin \theta_{qy}}=\frac{\sin \phi_q}{1}
\ee
Taking the modulus we therefore have
\be
\label{ineq}
|\sin \phi_x| \le |\sin \theta_{qy}|
\ee
Using the cosine rule for the spherical triangle $\triangle(q, x, z)$
\be
\cos \theta_{qz}= \sin \theta_{qx} \cos (\frac{\pi}{2}- \phi_x)=  \sin \theta_{qx} \sin \phi_x
\ee
Taking absolute values and substituting in (\ref{ineq}) then
\be
\label{uncert1}
|\sin \theta_{qx}| |\sin \theta_{qy}| \ge |\cos \theta_{qz}|
\ee
Now as discussed above, the mean value of the bits in the bit string on $S_z$ at co-latitude $\theta_{qz}$ equals $\cos \theta_{qz}$, whilst it is easily shown that $|\sin \theta_{qz}|$ equals the standard deviation. Writing $\cos \theta_{qz}= \bar S_z$, $\sin \theta_{qx}=\Delta S_x$ and $\sin \theta_{qy}=S_y$ then
\be
\label{uncert2}
\Delta S_x\Delta S_y \ge |\bar S_z|
\ee
If the elements of the bit string were set equal to $\pm\hbar/2$ instead of $\pm 1$, then (\ref{uncert2}) would be replaced by:
\be
\label{uncert3}
\Delta S_x\Delta S_y \ge \frac{\hbar}{2} |\bar S_z|
\ee
which is of course the expression of the Uncertainty Principle for qubits. 

But there is more to the Uncertainty Principle than this. Let $z^*$, and $x^*$ denote exact orientations inside the nominal orientations. If the cosines of the angular distances between $z^*$ and $x^*$ and between $q$ and $z^*$ are rational, then the cosine of the angular distance between $q$ and $x^*$ will not be rational. Similarly if we replace $x^*$ with $y^*$. That is to say, using the Impossible Triangle Corollary, if the discretised quantum state is defined at $q$ relative to a precise North Pole at $z^*$, it cannot simultaneously be defined relative with respect to a rotated North Pole at $x^*$ or $y^*$. 

Hence, this number-theoretic version of the Uncertainty Principle accounts for both the statistical relationship governing the mean and standard deviation of measurements on ensembles of particles, and the inherent contextuality of quantum spin. 

\subsection{Schr\"{o}dinger Evolution and Measurement.}
\label{evolution}

The distinct nature of unitary evolution and of measurement is straightforwardly described in this discretised model of quantum physics. Consider a quantum mechanical wavefunction prepared in state $|\psi(t_0)\rangle$, which evolves through a series of unitary transformations
\be
|\psi(t_n)\rangle = \mathcal U(t_j, t_{j-1}) \ \mathcal U(t_{j-1}, t_{j-2})\ldots \mathcal U(t_{2}, t_{1})\ \mathcal U(t_1, t_{0})
|\psi(t_{0})\rangle= \mathcal U(t_j, t_0) |\psi(t_0)\rangle
\ee
where $\mathcal U(t_j, t_{j-1})=e^{i H(t_j-t_{j-1})}$  and $H$ is the relevant Hamiltonian operator for the Schr\"{o}dinger equation. 

In the discretised model, the prepared state $|\psi (t_0))\rangle$ is associated with the bit string $\mathbbm 1$ at the north pole $z=z(t_0)$ of the Bloch sphere skeleton $S_{z(t_0)}$. $|\psi(t_1)\rangle$ is then associated with the bit string $i_1^{(m(t_1), n(t_1))}(\mathbbm 1)$ associated with the point $z(t_1)$ on $S_{z(t_0)}$ with colatitude/longitude  $(\theta(t_1), \phi(t_1))$, where $\cos^2 \theta(t_1)/2=1-m(t_1)/2N$ and $\phi(t_1)=2\pi n(t_1)/4N$. Continuing, we associate $|\psi(t_2)\rangle$ with the bit string 
\be
i_1^{(m(t_2), n(t_2))}( i_1^{(m(t_1), n(t_1))}(\mathbbm 1))
\ee
at the point $z(t_2)$ with colatitude/longitude $(\theta(t_2), \phi(t_2))$ relative to the Bloch sphere skeleton $S_{z(t_1)}$, where $\cos^2 \theta(t_2)/2=1-m(t_2)/2N$ and $\phi(t_2)=2\pi n(t_2)/4N$. In this way, $|\psi(t_n)\rangle$ is associated with the bit string
\be
\label{unitaries}
i_1^{(m(t_j), n(t_j))}(i_1^{(m(t_{j-1}), n(t_{j-1}))}( \ldots ( i_1^{(m(t_2), n(t_2))}(i_1^{(m(t_1), n(t_1))} (\mathbbm 1)))\ldots). 
\ee

We now come to the final measurement step. What distinguishes the measurement step from just another unitary step? The key point is that during each of the unitary steps, the bit strings are strictly ordered. Because of this ordering, it is always possible to undo $(\ref{unitaries})$ with inverse transformations, returning the ensemble state to $\mathbbm 1$. However, in the measurement step, the bit string is simply a disordered mixture.  

To make sense of this, we need to turn to a more geometric representation of these bit strings. To motivate such a geometric representation, consider first the possible objection that since the discretised model makes an ontological distinction between states with rational angles and states with rational cosines of angles, the model is unrealistically fine tuned. However, to make such an objection requires a specification of metric with respect to which the tuning is deemed fine. According to Ostrowsky's theorem \cite{Katok}, there exist two classes of norm-induced metric: the Euclidean metric and the $p$-adic metric. The set of $p$-adic integers is homeomorphic to Cantor sets with $p$ iterated pieces \cite{Katok}. This hints at the existence of some global fractal geometry $\mathcal I_U$ in state space to which states of the universe are constrained. Motivated by chaos theory, we call this the `invariant set' geometry in the sense that $\mathcal I_U$ is invariant under a dynamical evolution of states \cite{Palmer:2009a}. Points which lie off $\mathcal I_U$ are inconsistent with the geometric constraint and are $p$-adically distant from points on $\mathcal I_U$, no matter how close such points may seem to be from a Euclidean perspective. We can state this in a related way. Associated with $I_U$ is a non-trivial non-zero (Haar) measure. The states of the world which do not lie on $I_U$ have zero invariant measure \cite{Hance2022Supermeasured}. The geometric structure of $\mathcal I_U$ will be analysed further in Section \ref{evolution} when comparing unitary evolution and measurement from the discretised model perspective. 

Fig \ref{invariantset} illustrates the difference between unitary evolution and measurement from this geometric perspective. Fig \ref{invariantset} a) shows a state-space trajectory segment. At some low-level of iteration of $\mathcal I_U$ it appears a 1D curve. On magnification it is seen to comprise a helix of $4N$ trajectories. For simple unitary transformations where we write $e^{iH \Delta t}=e^{i \omega \Delta t}$, the unitaries (cyclic permutations of the bit strings) can be represented geometrically in terms of a helix of trajectories rotating with frequency $\omega$. In the relativistic Dirac equation, we would simply add a second helix rotating in the opposite direction (giving representation of the left- and right-handed Weyl spinors). 

A cross section through a trajectory segment would reveal a fractal structure homeomorphic to a Cantor set with $p$ iterated pieces, as shown in Fig \ref{invariantset} b. In Fig \ref{invariantset}c we associate measurement with a divergence of trajectories associated with the interaction of the quantum system with its environment, and a (nonlinear) clustering of trajectories into discrete clusters, which we label $1$ and $-1$.  If we form bit strings after the trajectories have clustered, then there is no natural order in the bit strings - they have become disordered. This can be contrasted with the helical structure in the unitary phase where the trajectories are ordered. Note that we can label the trajectories in the unitary phase of evolution in Fig \ref{invariantset}a with the clusters to which they subsequently evolve. This of course does not imply any kind of retrocausality. The cluster labels merely provide a convenient means to express the ordering of the trajectory segments. 

%..........................................................................................................................................
\begin{figure}
\centering
\includegraphics[scale=0.3]{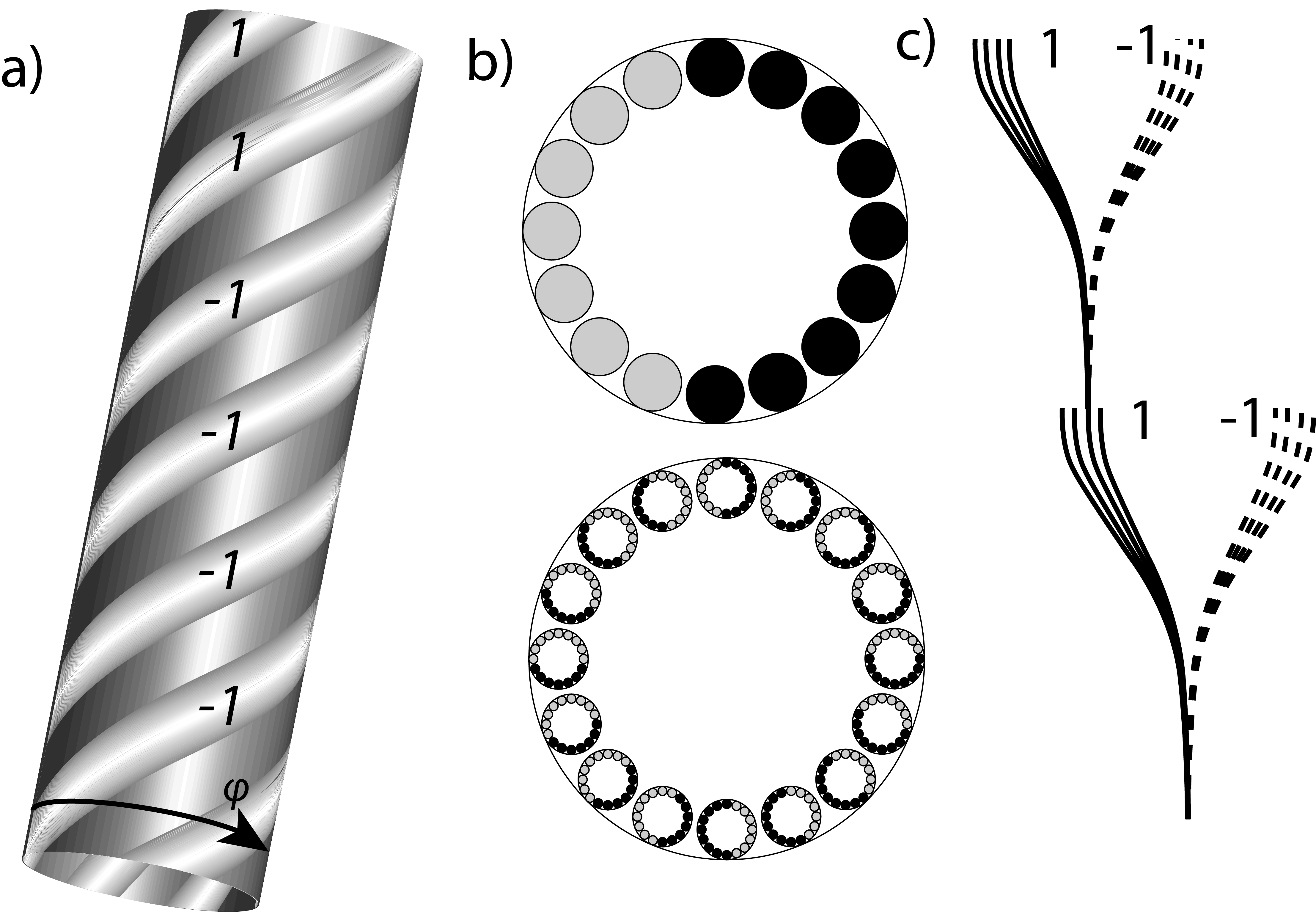}
\caption{\emph{Local state space geometry of the invariant set $I_U$. a) a state space trajectory which appears a 1D curve at a low iterate of the fractal geometry is found to comprise a helix of $4N$ trajectories (whose periodicity is related to the energy of the quantum system). b) A cross section through a trajectory is a Cantor set with $4N+n_X$ iterated pieces, homeomorphic to the $4N+n_X$-adic integers. c) On the invariant set, measurement is associated with a divergence of trajectories as a quantum system interacts with its environment, followed by a nonlinear (potentially gravitationally induced) clustering. In a) and b) each trajectory is labelled/coloured by the cluster to which it evolves in c).}}
\label{invariantset}
\end{figure}
%..........................................................................................................................................

This divergence of trajectories can be thought of as part of a `fractal zoom' and reveals the trajectory segments at the next level of iteration. In this sense the passage of time can itself be thought of as a fractal zoom of the invariant set $I_U$.  

These clusters are not the only way to label trajectories symbolically. For example, we can partition state space as the union of subspaces where a particle either went through the upper arm of a Mach-Zehnder device, or through the lower arm. We can form strings comprising bits which are the symbolic labels associate with such a partition. Hence, according to this version of the discretised model, a particle goes through the upper arm or the lower arm. It does not go through both arms simultaneously. From this perspective, the `non-classical' nature of quantum physics arises from the assumption that the laws of physics describe a geometry in state space. This geometry embodies information about the structure of neighbouring counterfactual trajectories where the particle went through the other arm of the interferometer. This is why it is necessary to invoke counterfactual computation \cite{MitchisonJozsa} to understand the power of quantum computers. By contrast, the laws of classical physics merely describe information about how systems evolve on a single (least action) trajectory. In quantum physics one cannot think of the trajectories in isolation as one can in classical physics. 

In this picture of measurement, the wavefunction does not `collapse' in the usual quantum sense of the word. By contrast, the wavefunction is always just an ensemble description of reality. On the other hand, the transition from the ordered bit string description to a disordered description (from helix to discrete clusters) clearly does signal some loss of information associated with what we call the `measurement process'. 

Above we have identified $p$ with the number $4N$ of Section \ref{IST}. However, for algebraic reasons, it is usually assumed that $p$ is prime. However, as discussed in Section \ref{IST}, there will exist trajectories which do not cluster, e.g., if they lie on the boundary between clusters. If we write $p=4N+n_X$ where $n_X$ denotes a number of such null states, then it becomes possible for $p$ to be prime. If $n_X =1$, indicating a single null trajectory at the boundary between the partition into $1$ and $-1$, then $p$, if prime, would be a Pythagorean prime. The $p$-adic theory underpinning the symbolic bit string representations of the quantum state has not been developed in sufficient detail to be able to make more precise observations of the relationship between $p$ and $N$ than this. However, numerical estimates of $N$ and hence $p$ are made in Section \ref{sizeofN}. 

With this geometric link established, the discretised model is henceforth referred to by the name `invariant set theory' \cite{Palmer:2020}. 

\subsection{The Quantum Mechanical and Classical Limits of Invariant Set Theory and the size of $N$}
\label{sizeofN}

The quantum mechanical limit of invariant set theory occurs when $N$ is set equal to infinity. At this limit, the whole of complex Hilbert space is part of the state space of the theory, the rationality conditions become nullified and we are simply left with the state space of quantum mechanics. However, it is important to understand that $N=\infty$ is a singular limit. Invariant set theory does not steadily approach quantum mechanics as $N$ gets bigger and bigger, since the constraints of Niven's theorem holds no matter how large is $N$. Berry \cite{Berry} has noted that singular limits are commonplace in physics and that frequently an old theory is a singular limit of a new theory as a parameter of the new theory is set equal to zero or infinity. 

How large is $N$ if is not infinite? One can start by considering a minimum size $\epsilon_0$ for an $\epsilon$-disk. We know that attempts to measure a quantum system on scales smaller than the Planck scale will require such a concentration of energy as to form a black hole. This suggests that $\epsilon_0$ and hence the largest value $N$ could be is set by gravity. This in turn suggests that $N$ is a function of the energy $E$ of the quantum system under consideration, i.e., $N \sim E/ E_{\text{Planck}}$. For an infrared photon this would give $N \sim10^{26}$ and for a quantum system with Compton wavelength equal to the diameter of the observable universe ($N \sim10^{62}$). 

The classical limit of invariant set theory occurs at $N=1$, when the fractal structure of trajectories disappears. Then we simply revert to single state-space trajectories without internal structure. This would occur for quantum systems where $E \ge E_{\text{Planck}}$, i.e. for particles whose mass exceeds about 22 $\mu$g (for comparison, a tardigrade has mass about 10 $\mu$g). In this sense, invariant set theory predicts that quantum systems with detectable self gravitation will not show quantum mechanical properties. The notion that $N$ (and hence $p$) is in part set by the strength of gravity is consistent with the conjecture that the clustering of trajectories on $I_U$, associated with the measurement process, is itself a gravitational phenomenon \cite{Diosi:1989} \cite{Penrose:2004}.  

In Section \ref{K}, it was noted that a general $K$ qubit state was described by $K$ bit strings of length $2^{K+1}N$. If we take the largest value of $2^{K+1}N$ to be $10^{62} \sim 2^{205}$, and the smallest value of $N$ as 1, then this will imply a maximum value $K_{\text{max}}$ of $K$ to be around $200$ (before the degrees of freedom become restricted by the maximum length of a bit string). Of course the size of the observable universe is not the same as the size of the universe, and we do not even know if the latter is finite. However, if it is finite, invariant set theory predicts that the exponential increase in the number of freedom will eventually cease. Some estimates based on Bayesian model averaging suggests a largest size at least some 250 greater than the observable size. However, a factor 250 only increases $K_{\text{max}}$ by about 8. Hence, depending on the size of the universe, it may not be possible to build a general purpose quantum computer with more than a few hundred qubits. 

\section{Entangled qubits and nonlocality}
\label{entangled}

\subsection{Bell's Theorem}
\label{Bell}

In quantum mechanics, the Bell state is given by 
\be
\label{bell}
|\psi_{\text{Bell}}\rangle = \frac{1}{\sqrt 2} \big( |\ 1\rangle |\ -1\rangle \  +\  |-1\rangle |1\rangle \big) 
\ee
In invariant set theory, consistent with Section (\ref{K}), $|\psi_{\text{Bell}}\rangle $ is associated with a pair of bit strings of length 8N (one for experimenter Alice and one for experimenter Bob)
\begin{align}
\label{alicebobbits}
\mathcal B_a&=\ \ i_1^{(m_a)} (\ \ \mathbbm 1)\ || \  i_1^{(m_a)}(-\mathbbm 1) 
\nonumber \\
\mathcal B_b&= \ \ i_1^{(m_b)} (-\mathbbm 1)\ || \  i_1^{(m_b)}(\ \ \ \mathbbm 1)
\end{align}
From (\ref{alicebobbits}) it is straightforwardly seen that $\mathcal B_a$ and $\mathcal B_b$ contain equal numbers of $1$ and $-1$ bits for all values of the free parameters $m_a$ and $m_b$. In addition, the correlation between $\mathcal B_a$ and $\mathcal B_b$ is equal to $|m_b-m_a|/N-1$. As such, we can represent $\mathcal B_a$ and $\mathcal B_b$ as a pair of points on a 2-sphere whose angular separation is $\theta_{ab}$ such that $\cos \theta_{ab}=1- |m_b-m_a|/N \in \mathbb Q$. These two points will correspond to the (exact) orientations of the experimenters' Alice and Bob's measuring apparatuses. 

For each run of a CHSH experiment, the experimenters Alice and Bob choose one of two possible spin measurements to perform, denoted by $x$ and $y$, where $x, y \in \{0,1\}$. Such $x$ and $y$ define the nominal orientations of SG devices/polarisers. As before, we treat these choices $x=0$, $x=1$, $y=0$ and $y=1$ as four small $\epsilon$-disks on the 2-sphere of measurement orientations. We label Alice and Bob's measurement outcomes as $a$ and $b$ where $a, b \in \{1,-1\}$. From (\ref{alicebobbits}), for any run of the experiment there exist a pair of exact measurement orientations whose relative orientation satisfies $\cos \theta_{ab}= 1-|m_b-m_a|/N$. As before, although the discretised model requires $\cos \theta_{ab}$ to be rational, Alice and Bob do not have any \emph{direct} control over the rationality or irrationality of $\cos \theta_{ab}$. Rather rationality is inevitably associated with the model's representation of the real world. 

Since we can represent quantum mechanical Bell states to arbitrary precision in invariant set theory, then, for large enough $N$, the violation of Bell's inequality is as close to quantum mechanics as we like. In the rest of this section we show that this comes about from a violation of the Statistical Independence (SI) assumption 
\be
\label{SI}
\rho(\lambda | x \; y)=\rho(\lambda)
\ee
where $\rho(\lambda)$ is a probability density on the hidden variables. In particular, it does \emph{not} come about from a violation of local causality
\be
\label{locality}
a=A(\lambda, x);\ \ \ b=B(\lambda, y)
\ee
where $A$ and $B$ denote deterministic functions and $\lambda$ hidden variables. The locality condition expresses the idea that the outcome of Alice's measurements do not depend on Bob's choices or vice versa. As discussed, the violation of Statistical Independence is non-conspiratorial, that is to say, it does not depend on Alice and Bob making `just the right choices'. Instead there is nothing in the analysis below to prevent Alice and Bob from chosing as they like. In this sense the hidden variables are not akin to some `alien mind control' as some observers imagine a violation of SI to be. 

From (\ref{locality}), the hidden variables are those which help determine measurement outcomes, but over which the experimenters have no direct control. Following the analysis of the Mach-Zehnder and the Stern-Gerlach experiments, the hidden variable of Alice's particle is associated with the exact measurement setting $x^*$ relative to the $\epsilon$-disk $x$ (the hidden variable associated with Bob's particle is associated with the exact setting $y^*$ relative to $y$). 

%..........................................................................................................................................
\begin{figure}
\centering
\includegraphics[scale=0.5]{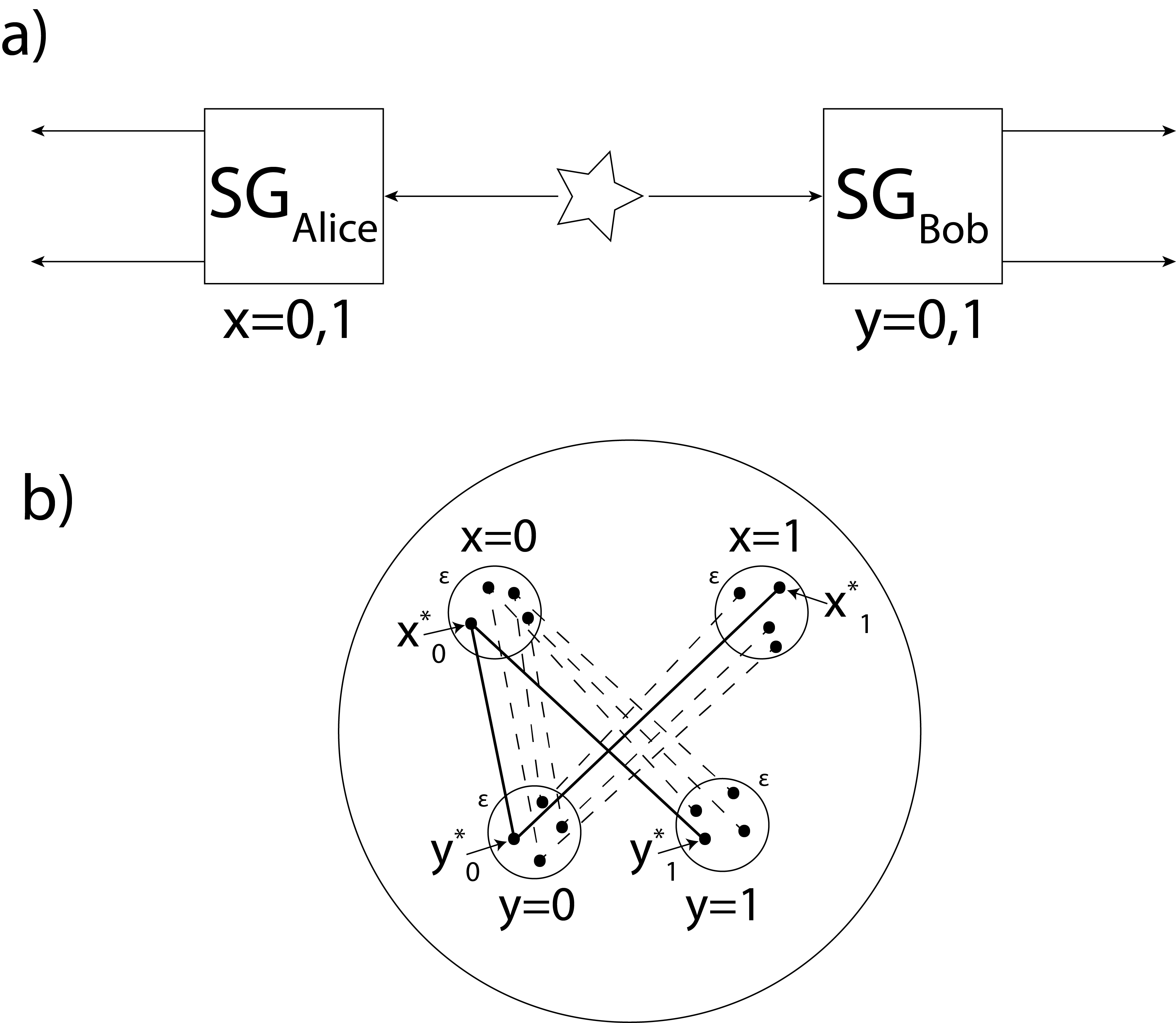}
\caption{\emph{a) An schematic set up for the CHSH Bell experiment with nominal measurement settings $x \in \{0,1\}$ for Alice and $y \in \{0,1\}$ for Bob. b) For a particular run, Alice chooses $x=0$ and Bob $y=0$. The points $x^*_0$, $y^*_0$, $x^*_1$, $y^*_1$, one within each of the $\epsilon$ disks, denote exact settings for some particular value of $\lambda$. The cosines of the angular distances between $x^*_0$ and $y^*_0$, between $x^*_0$ and $y^*_1$, and between $x^*_1$ and $y^*_0$, are all rational. However, given this, the cosine of the angular distance between $x^*_1$ and $y^*_1$ cannot be rational. Hence, with $\lambda$ fixed (i.e., $x^*_0$, $y^*_0$, $x^*_1$, $y^*_1$ fixed), $0=\rho(\lambda | x=1\; y=1) \ne \rho(\lambda | x=0 \; y=0) \ne 0$ and Statistical Independence (\ref{SI}) is violated. Alice's exact measurement setting $x^*_0$ and hence measurement outcome (when she chooses $x=0$) does not depend on whether Bob chooses $y=0$ or $y=1$. Hence the model satisfies the locality condition (\ref{locality}). 
}}
\label{CHSHfig}
\end{figure}
%..........................................................................................................................................

Fig \ref{CHSHfig}b illustrates a particular run of the experiment (particular $\lambda$) where Alice chooses $x=0$ and Bob $y=0$. According to invariant set theory there exist pairs of points inside the $\epsilon$-disks associated with $x=0$ and $y=0$, such that the cosine of the angular distance between the exact points is rational. There will be as many such points as we like for large enough $N$. Let $x^*_0$ and $y^*_0$ denote two possible exact points. Hence $\rho(\lambda |x=0 \; y=0) \ne 0$ in the model, consistent with what happens in the real world, where $\lambda$ is associated with the pair $\{x^*_0, y^*_0\}$. 

Consider now the counterfactual world where, for the same particle pair, i.e., the same $\lambda$, Alice chose $x=0$, as she did in the real world, but Bob chose $y=1$. If SI holds, then $\rho(\lambda |x=0 \; y=1)=\rho(\lambda |x=0 \; y=0)\ne 0$. According to the discussion above, fixing $\lambda$ \emph{and} fixing $x=0$ means fixing the exact point $x^*_0$ in the $\epsilon$-disk $x=0$. If $\rho(\lambda |x=0 \; y=1) \ne 0$, then fixing $\lambda$ and varying $y=0$ to $y=1$, there must exist an exact point $y^*_1$ in the $\epsilon$ disk $y=1$ such that the cosine of the angular distance between $x^*_0$ and $y^*_1$ is rational. As before there are many such possibilities, Fig \ref{CHSHfig} shows one choice for $y^*_1$. Hence, for Bob's counterfactual choice $y=1$ we can additionally associate $\lambda$ with the exact orientation $y^*_1$, and hence $\lambda$ is now associated with the triple $\{x^*_0, y^*_0, y^*_1\}$. With this, $\rho(\lambda |x=0 \; y=1) \ne 0$ consistent with (\ref{SI}). 

Consider now a second counterfactual world where, again keeping $\lambda$ fixed, Alice chose $x=1$ and Bob $y=0$. Keeping $\lambda$ fixed, we keep fixed the exact orientation $y^*_0$ in the $\epsilon$-disk $y=0$. If $\rho(\lambda |x=1 \; y=0) \ne 0$, then fixing $\lambda$ and varying $x=0$ to $x=1$, there must exist an exact orientation $x^*_1$ in the $\epsilon$ disk $x=1$ such that the cosine of the angular distance between $y^*_0$ and $x^*_1$ is rational. As before there are many such possibilities, Fig \ref{CHSHfig}b shows one choice for $x^*_1$. Hence, for Alice's counterfactual choice $x=1$ we can associate $\lambda$ with the exact choice $x^*_1$ and hence $\lambda$ is now associated with the quadruple $\{x^*_0, x^*_1, y^*_0, y^*_1\}$. With this, $\rho(\lambda |x=1, y=0) \ne 0$ consistent with (\ref{SI}). 

So far SI has not been violated. However, finally, consider the third counterfactual world where, keeping $\lambda$ fixed, Alice chose $x=1$ and Bob $y=1$. Now both of the exact orientations inside the $\epsilon$ disks $x=1$ and $y=1$ have already been fixed (see Fig \ref{CHSHfig}b). Hence we have no free choices left. If $\rho(\lambda |x=1 \; y=1) \ne 0$, it must be the case that the cosine of the angular distance between $x^*_1$ and $y^*_1$ is rational. However, this cannot be the case as a straightforward generalisation of the Impossible Triangle Corollary shows. 

\begin{corollary} Let $x^*_0$, $x^*_1$, $y^*_0$, $y^*_1$ denote 4 exact points inside the $\epsilon$ disks, $x=0$, $x=1$, $y=0$, $y=1$ respectively, such that $\cos \theta_{x^*_0 y^*_0}$, $\cos \theta_{x^*_0 y^*_1}$ and $\cos \theta_{x^*_1 y^*_0}$ are rational. If $\phi_{x^*_0}$ and $\phi_{x^*_1}$ are independent rational angles subtended by the great circles at $x^*_0$ and $x^*_1$ then $\cos \theta_{x^*_1 y^*_1} \notin \mathbb Q$. 
\end{corollary}

\begin{proof}
As before, we give a proof by contradiction and assume $\cos \theta_{x^*_1 y^*_1}$ is rational. Using the cosine rule for the two triangles $\triangle(x^*_0y^*_0y^*_1)$ and $\triangle(x^*_1y^*_0y^*_1)$, we have two expressions for $\cos \theta_{y^*_0 y^*_1}$:
\begin{align}
\cos \theta_{y^*_0 y^*_1}&=\cos \theta_{x^*_0 y^*_0} \cos \theta_{x^*_0 y^*_1}+\sin \theta_{x^*_0 y^*_0} \sin \theta_{x^*_0 y^*_1} \cos \phi_{x^*_0} \nonumber \\
\cos \theta_{y^*_0 y^*_1}&=\cos \theta_{x^*_1 y^*_1} \cos \theta_{x^*_1 y^*_0}+\sin \theta_{x^*_1 y^*_1} \sin  \theta_{x^*_1 y^*_0} \cos \phi_{x^*_1}
\end{align}
If $\cos \theta_{x^*_0 y^*_0}$, $\cos \theta_{x^*_0 y^*_1}$, $\cos \theta_{x^*_1 y^*_0}$ and $\cos \theta_{x^*_1 y^*_1}$ are all rational, then, taking the difference between these equations, 
\be
\label{diff}
\sin \theta_{x^*_0 y^*_0} \sin \theta_{x^*_0 y^*_1} \cos \phi_{x^*_0}-\sin \theta_{x^*_1 y^*_1} \sin  \theta_{x^*_1 y^*_0} \cos \phi_{x^*_1}
\ee
must be rational. If $\phi_{x^*_0}$ and $\phi_{x^*_1}$ are independent angles, then (\ref{diff}) can only be assumed rational if the two individual summands are rational. Squaring the individual summands then we deduce that $\cos 2\phi_{x^*_0}$ and $\cos 2\phi_{x^*_1}$ must both be rational. As discussed in the Impossible Triangle Corollary, this contradicts the assumption that $\phi_{x^*_0}$ and $\phi_{x^*_1}$ are rational angles. Hence $\cos \theta_{x^*_1 y^*_1} \notin \mathbb Q$. 
\end{proof}
Hence $\rho(\lambda |x=1 \; y=1) = 0$ and SI is violated. 

We now show that this model satisfies locality (\ref{locality}).  In Fig \ref{CHSHfig}b) it can be seen that the exact measurement setting $x^*_0$ - the setting associated with Alice's particle's hidden variable - does not depend on whether Bob chose $y=0$ or $y=1$; either choice is allowed for a fixed $x^*_0$ as evidenced by the two solid lines emanating from $x^*_0$. Similarly, the exact measurement setting $y^*_0$ - Bob's particle's hidden variable - does not depend on whether Alice chose $x=0$ or $x=1$; either choice is allowed for a fixed $y^*_0$. Hence, the measurement outcome $a$ does not depend on Bob's choice $y$, nor $b$ depend on Alice's choice $x$. This is a statement that local causality holds. 

More generally, if $\bar x=1-x$, $\bar y=1-y$, then 
\be
\rho(\lambda |x, y) \ne 0;\ \  \rho(\lambda |x, \bar y) \ne 0;\ \ \rho(\lambda |\bar x, y) \ne 0
\ee
implies
\be
\rho(\lambda | \bar x, \bar y)=0
\ee
It is worth noting that the proof of the violation of Statistical Independence in the original Bell inequality \cite{Bell:1964} can be derived directly from the Impossible Triangle Corollary and does not require the extension above. In particular, the three $\epsilon$-disks in Fig \ref{SG} c) correspond to the three nominal settings in the original Bell inequality. Exactly the same argument applies. 

As mentioned, there is nothing to constrain Alice and Bob's choices of \emph{nominal} settings. They can chose as they like (e.g, on the wavelength of light from distant quasars, or from the Dow Jones index, or any whimsical alternative). However, the exact measurement settings and hence the hidden variables are constrained by these choices. This does not occur in classical mechanics: throwing a stone with a certain nominal speed does not limit the set of exact solutions to the governing equations where the stone was thrown faster or slower. Importantly, this constraint on exact settings does not imply any breakdown of local causality in space-time - once again, the model satisfies the locality condition (\ref{locality}). Quantum physicists may wish to use the word `nonlocality' to describe this global state-space geometric constraint. However, it should not be conflated with some unattractive and implausible action-at-a-distance which was the basis of Einstein's dislike of quantum mechanics. 

A model which violates Statistical Independence is sometimes  called `superdeterministic' \cite{HossenfelderPalmer} though here the term `supermeasured' \cite{Hance2022Supermeasured} is to be preferred as this refers to the notion that the dynamics must respect the invariant measure associated with $I_U$. This violation of Statistical Independence does not imply some kind of fine-tuned conspiracy. For one thing, as mentioned, Alice and Bob can chose as they like. In addition, as discussed above, the natural metric on the state space of this discretised model is $p$-adic and bit strings which do not satisfy the rationality conditions are $p$-adically distant from those that do. This is not a conspiracy; it is an elegant interpretation of an otherwise problematic experiment in quantum foundations. 

\subsection{GHZ}

The discretised model also provides a realistic description of the entangled GHZ state \cite{GHZ}, making use of the number-theoretic incommensurateness relevant to the Mach-Zehnder interferometer as discussed in Section \ref{MachZ}. Here
\be
|\psi_{\text{GHZ}}\rangle=\frac{1}{\sqrt 2}(|v_A\rangle|v_B\rangle |v_C\rangle +|h_A\rangle|h_B\rangle |h_C\rangle)
\ee
describes the entangled state of three photons $A$, $B$ and $C$, in a basis where $v$ and $h$ denote vertical and horizontal polarisation. Through the unitary transformation 
\be
\label{GHZtrans}
\begin{pmatrix}
v' \\
h'
\end{pmatrix}
=
\begin{pmatrix}
\cos \phi & - \sin \phi \\
\sin \phi & \cos \phi
\end{pmatrix}
\begin{pmatrix}
v \\
h
\end{pmatrix}
\ee
linear polarisation measurements can be made on any one of the three photons at angle $\phi$ to the $v/h$ axis. This requires $\cos^2 \phi$ and hence $\cos 2 \phi$ to be rational. Conversely, through the unitary transformation
\be
\label{GHZtrans2}
\begin{pmatrix}
L\\
R
\end{pmatrix}
=
\begin{pmatrix}
1 & -i \\
1 & i
\end{pmatrix}
\begin{pmatrix}
v \\
h
\end{pmatrix}\ \ \ \ \ \ \ \ \ \ 
\ee
it is possible to make circular polarisation measurements on the other two photons.

The resulting correlations are impossible to explain with a hidden variable theory if we assume that the first photon could have been measured relative to the $L/R$ basis. However, from (\ref{GHZtrans}) and (\ref{GHZtrans2})
\begin{align}
\begin{pmatrix}
L\\
R
\end{pmatrix}
=&
\begin{pmatrix}
1 & -i \\
1 & i
\end{pmatrix}
\begin{pmatrix}
\cos \phi & \sin \phi \\
-\sin \phi & \cos \phi
\end{pmatrix}
\begin{pmatrix}
v' \\
h'
\end{pmatrix} \nonumber \\
=&
\begin{pmatrix}
e^{i \phi} & e^{i(\phi-\pi/2)} \\
e^{-i\phi} & e^{-i(\phi-\pi/2)}
\end{pmatrix}
\begin{pmatrix}
v' \\
h'
\end{pmatrix}
\end{align}
As in the discussion in Section \ref{MachZ}, there is a number theoretic inconsistency which prevents linear and circular polarisation being simultaneously measured: $\phi$ can be a rational angle, or it can have a rational cosine. 

We do not have to give up deterministic realism to explain GHZ in this discretised picture of quantum physics. 

\section{Discussion}
\label{discussion}

This work was motivated by a belief that the synthesis of quantum and gravitational physics will require more give from the quantum side and less from the gravitational side. Here we have `second quantised' quantum mechanics in an unconventional way: by discretising complex Hilbert space. The particular discretisation leads to an inherent number-theoretic incommensurateness between complex amplitudes and complex phases. The discretisation allows complex Hilbert states to be given an ensemble interpretation, and the incommensurateness generates all the conventional properties of quantum systems typically associated with the algebra of Hermitian operators in complex Hilbert Space, but without invoking such an algebra. Importantly, with this discretisation, it is no longer necessary to interpret the violation of Bell inequalities as implying a breakdown of causal realism in space-time, making the discretised model more consistent with the causal determinism of general relativity than is quantum mechanics. 

Although the experimental consequences of this discretised theory can be made identical to quantum physics by setting the parameter $N$ to a large enough value, here estimates of $N$ have been given that lead to experimentally testable differences from quantum mechanics: that systems with masses greater the Planck mass will \emph{not} exhibit quantum properties, and that general purpose noise-free quantum computers may in principle be limited to a few hundred qubits. The latter invokes the fact that the discretised theory has a geometric underpinning, describing the fractal invariant set of the universe. 

There is one further implication of this fractal geometry worth commenting on here. Consider a state space trajectory at some low level of fractal iteration and zoom into it with higher and higher magnification. At the next level of fractal iteration the trajectory appears as a helix of trajectories, i.e., as an ensemble. At the next level of iteration the trajectory again returns to appear as a single curve. A trajectory vacillates between a single curve and a helix ensemble as the zoom proceeds. This means that the laws of physics similarly should vacillate between being classical-like and quantum-like. This gives a very different picture to the usual one, where classical theories are assumed to be less fundamental than quantum theories. The author intends to expand on this in more detail elsewhere.    

Although a radical modification of quantum mechanics has been proposed, there has to be some modification to general relativity to make it compatible with this discretised version of quantum physics. Most importantly, the energy-momentum tensor on the right hand side of the field equations should describe not only energy-momentum distributions in our space time $\mathscr M$, but energy-momentum distributions on space-times $\mathscr M'$ in the neighbourhood of our space time, on the $p$-adic geometry in state space, suitably weighted with a suitable Haar measure $\rho_H$, i.e.,
\be
G_{\mu\nu}(\mathcal M)=8\pi\int_{\mathcal M'} d\mathcal  M' \  \rho_H({\mathcal  M, \mathcal M'})\ T_{\mu\nu}(\mathcal M')
\ee
In the classical ($N=1$) limit of GR, $\rho_H({\mathcal M, \mathcal M'})=\delta(\mathcal M - \mathcal M')$. Whether the difference between $\rho_H({\mathcal M, \mathcal M'})$ and $\delta(\mathcal M - \mathcal M')$ can help explain the existence of dark matter (as regular matter on neighbouring space-times on $I_U$, having a gravitational effect in our space-time) remains to be seen. 

\section*{Acknowledgements} My thanks to Michael Hall, Jonte Hance and Sabine Hossenfelder for helpful discussions. 
\bibliography{mybibliography}
\end{document}